\documentclass[runningheads]{llncs}
\usepackage{mathpartir}
\usepackage{mathtools}
\usepackage{color}
\usepackage{tikz}
\usetikzlibrary{calc,trees,positioning,arrows,chains,shapes.geometric,%
	decorations.pathreplacing,decorations.pathmorphing,shapes,%
	matrix,shapes.symbols}

\tikzset{
	>=stealth',
	punktchain/.style={
		rectangle, 
		rounded corners, 
		draw=black, very thick,
		text width=15.5em, 
		minimum height=2em, 
		text centered, 
		on chain},
	line/.style={draw, thick, <-},
	element/.style={
		tape,
		top color=white,
		bottom color=blue!50!black!60!,
		minimum width=8em,
		draw=blue!40!black!90, very thick,
		text width=10em, 
		minimum height=3.5em, 
		text centered, 
		on chain},
	every join/.style={->, thick,shorten >=1pt},
	decoration={brace},
	tuborg/.style={decorate},
	tubnode/.style={midway, right=2pt},
}

\usepackage{xcolor}
\newcommand\hl[1]{\textcolor{black}{#1}}

\usepackage{graphicx}

\definecolor{pblue}{rgb}{0.13,0.13,1}
\definecolor{pgreen}{rgb}{0,0.5,0}
\definecolor{pred}{rgb}{0.9,0,0}
\definecolor{pgrey}{rgb}{0.46,0.45,0.48}
\definecolor{javapurple}{rgb}{0.5,0,0.35}

\usepackage{listings}
\newcommand\Q\lstinline
\newcommand{\hoare}[3]{\mathtt{\{#1\} \; #2 \; \{#3\}}}
\newcommand{\E}{\ensuremath{\mathcal{E}}}
\newcommand{\TD}{\ensuremath{\mathit{TD}}}
\newcommand{\CD}{\ensuremath{\mathit{CD}}}
\newcommand{\Body}{\ensuremath{\mathit{Body}}}
\newcommand{\MH}{\ensuremath{\mathit{MH}}}
\newcommand{\vs}{\ensuremath{\mathit{vs}}}
\newcommand{\es}{\ensuremath{\mathit{es}}}
\newcommand{\Cs}{\ensuremath{\mathit{Cs}}}
\newcommand{\Ds}{\ensuremath{\mathit{Ds}}}
\newcommand{\Ms}{\ensuremath{\mathit{Ms}}}
\newcommand{\Name}{\ensuremath{\mathit{Name}}}
\newcommand{\OK}{\ensuremath{\mathit{OK}}}
\newcommand{\Pre}{\ensuremath{\mathit{Pre}}}
\newcommand{\Post}{\ensuremath{\mathit{Post}}}

\newcommand{\newKW}{\ensuremath{\mathtt{new}}}
\newcommand{\methodKW}{\ensuremath{\mathtt{method}}}
\newcommand{\result}{\ensuremath{\mathtt{result}}}
\newcommand{\this}{\ensuremath{\mathtt{this}}}
\newcommand{\interface}{\ensuremath{\mathtt{interface}}}

\usepackage{multirow}
\usepackage{multicol}
\usepackage{booktabs}
\usepackage{tabularx}
\newcolumntype{b}{X}
\newcolumntype{n}{>{\hsize=.35\hsize}X}
\newcolumntype{s}{>{\hsize=.2\hsize}X}
\newcolumntype{Y}{>{\centering\arraybackslash}X}
\makeatletter
\newcommand\cellwidth{\TX@col@width}
\makeatother
\makeatletter
\let\tx@\TX@endtabularx
\def\restoretx{\let\TX@endtabularx\tx@}
\makeatother

\usepackage{ltablex}

\begin{document}
\renewcommand{\thelstlisting}{\arabic{lstlisting}}

\title{Traits for Correct-by-Construction Programming}
\titlerunning{Traits for Correct-by-Construction Programming}

\author{Tobias Runge\inst{1,2} \and Alex Potanin\inst{3} \and Thomas Thüm\inst{4} \and Ina Schaefer\inst{1,2}}
\institute{TU Braunschweig, Germany
	\and Karlsruhe Institute of Technology, Germany
	\and Australian National University, Australia
	\and University of Ulm, Germany\\
	\email{\{tobias.runge,ina.schaefer\}@kit.edu, alex.potanin@anu.edu.au,\\ thomas.thuem@uni-ulm.de}
}
\authorrunning{Tobias Runge et al.}

\maketitle

\begin{abstract}
We demonstrate that traits are a natural way to support correctness-by-construction (CbC) in an existing programming language in the presence of traditional post-hoc verification (PhV).
With Correct- ness-by-Construction, programs are constructed incrementally along with a specification that is inherently guaranteed to be satisfied. CbC is complex to use without specialized tool support, since it needs a set of refinement rules of fixed granularity which are additional rules on top of the programming language.

In this work, we propose TraitCbC, an incremental program construction procedure that implements correctness-by-construction on the basis of PhV by using traits. TraitCbC enables program construction by trait composition instead of refinement rules. It provides a programming guideline, which similar to CbC should lead to well-structured programs, and allows flexible reuse of verified program building blocks.
We introduce TraitCbC formally and prove the soundness of our verification strategy. Additionally, we implement TraitCbC as a proof of concept.


\end{abstract}

\section{Introduction}
\emph{Correctness-by-Construction} (CbC)~\cite{dijkstra-book,gries-book,kourie2012correctness,morgan-book} is a methodology that incrementally constructs correct programs guided by a pre-/postcondition specification.\footnote{The approach should not be confused with other CbC approaches such as CbyC of Hall and Chapman~\cite{hall_chapman_2002}. CbyC is a software development process that uses formal modeling techniques and analysis for various stages of development (architectural design, detailed design, code) to detect and eliminate defects as early as possible~\cite{Chapman:2006}. \hl{We also exclude data refinement from abstract data types to concrete ones during code generation as for example in Isabelle/HOL~\cite{haftmann2013data}}.}
CbC uses small tractable refinement rules where 
in each refinement step, an abstract statement (i.e., a hole in the program) is refined to a more concrete implementation that can still contain some nested abstract statements. While refining the program, the correctness of the whole program is guaranteed through the check of conditions in the refinement rules.
The construction ends when no abstract statement is left.
Through the structured reasoning discipline that is enforced by the refinement rules, it is claimed that
program quality increases and verification effort is reduced~\cite{kourie2012correctness,watson2016correctness}.

Despite these benefits, CbC has a drawback: the refinement rules extend the programming language (i.e., refinements are an additional linguistic construct to transform programs). Special tool support~\cite{runge2019tool} is necessary to introduce the CbC refinement process to a programming language.
Additionally, the predefined rules have a fine granularity such that for every new statement the programmer adds to the program, an application of a refinement rule is necessary. Consequently, the concepts of CbC (e.g., abstract statements and refinement rules) increase the effort and necessary knowledge of the developer to construct programs.

\emph{Post-hoc verification} (PhV) is another approach to develop correct programs. A method is verified against its pre- and postconditions after implementation.
In practice, it often happens that a program is constructed first, with the objective of verifying it later~\cite{watson2016correctness}.
This can lead to tedious verification work if the program is not well-structured.
An example is the difficult search for the many reasons preventing the verification of a method to be completed:
an incorrect specification, an incorrect method, or inadequate tool support. Therefore, a structured programming approach is desirable to construct programs which are amenable to software verification.

In this work, we use \emph{traits}~\cite{ducasse2006traits} to overcome the drawbacks of CbC (complex programming style using external refinement rules) and introduce a programming guideline for an incremental trait-based program construction approach that guarantees that the resulting trait-based program is correct-by-construction. TraitCbC is based on PhV. With TraitCbC, the same programs can be verified as with PhV, but in addition, TraitCbC introduces an explicit program construction approach. It utilizes the flexibility of traits, which is beneficial for scenarios as incremental development~\cite{damiani2014verifying} and the development of software product lines~\cite{CN02,bettini2010implementing}.

\emph{Traits}~\cite{ducasse2006traits} are a flexible object-oriented language construct supporting a rich form of modular code reuse orthogonal to inheritance.
A trait contains a set of \emph{concrete} or \emph{abstract} methods (i.e., the method has either a body or has no body), independent of any class or inheritance hierarchy.\footnote{The term trait has been used by many programming languages: Java interfaces with default methods are a good approximation for what has been called trait in the literature, while Scala traits are mixins~\cite{flatt1998classes}, and Rust traits are type classes~\cite{sozeau2008first}.}  
Traits are independent modules that can be composed into larger traits or classes. When traits are composed, the resulting code contains all methods of all composed traits.
To verify traits, Damiani et al.~\cite{damiani2014verifying} proposed a modular and incremental \emph{post-hoc verification} process. Each method in every trait is verified in isolation by showing that the method satisfies its \emph{contract}~\cite{meyer1988eiffel}. Then, during the composition of traits, it has to be checked whether a method implemented in one trait is compatible with the abstract method with the same signature in another trait. That means, a concrete method has to satisfy the specification of the abstract method.
A concrete method with a weaker precondition and a stronger postcondition fulfills the contract of the abstract method (cf. Liskov substitution principle~\cite{liskov}).

A developer using TraitCbC starts by implementing a method (e.g., a method \texttt{a}) in a first trait. Similar to CbC, the method can contain holes that are refined in subsequent steps. A hole in TraitCbC is an abstract method (e.g., an abstract method \texttt{b}) that is called in method \texttt{a}; that is, a call to an abstract method corresponds to an abstract statement in CbC.
In the next step, one of these new abstract methods (e.g., \texttt{b}) is implemented in a second trait, again more abstract methods can be declared for the implementation. Similar to PhV, it must be proven that the implemented methods satisfy their specifications. Afterwards, the traits are composed; the composition operation checks that the contract of the concrete method \texttt{b} in the second trait fulfills the contract of the abstract method \texttt{b} in the first trait. This incremental process stops when the last abstract method is implemented, and all traits are composed. 

The main result of our work is the discovery that traits intrinsically enable correctness-by-construction. This work is not about pushing verification forward in the sense of adding more expressive power.
TraitCbC realizes a refinement-based program development approach using pre-/postcondition contracts and method calls instead of refinement rules and abstract statements as in CbC.
Refinement rules in the form of trait composition exist as a direct concept of the programming language instead of being a program transformation concept.
Additionally, each method implemented in the refinement process can be reused by composing traits in different contexts (i.e., already proven methods can be called by new methods under construction). This is advantageous compared to the limited reuse potential of methods in class-based inheritance.
Finally, TraitCbC is parametric w.r.t. the specification logic.
Thus, a language with traits can adopt the proposed CbC methodology.




\section{Motivating Example}
\label{sec:example}
In this section, we go through an example of how our development process enables CbC using traits.


\paragraph{Incremental Construction of MaxElement}

We use a sample object-oriented language in the code examples.
We construct a method  \Q@maxElement@ that finds the maximum element in a list of numbers.
A list has a head and a tail.
Only non-empty lists have a maximum element. This is explicit in the precondition of our specification, where we require that the list has at least one element. In the postcondition, we specify that the result is in the list and larger than or equal to every other element. A method \Q@contains@ checks that the result is a member of the list.
In the first step, we create a trait \Q@MaxETrait1@ that defines the abstract method \Q@maxElement@. 
The method \Q@maxElement@ is abstract, i.e., equivalent to an abstract statement in CbC.

\begin{lstlisting}
trait MaxETrait1 {
  @Pre: list.size() > 0
  @Post: list.contains(result) &
    (forall Num n: list.contains(n) ==> result >= n)
  abstract Num maxElement(List list);
}
\end{lstlisting}

In the second step in trait \Q@MaxETrait2@, we implement the method \Q@maxElement@ using two abstract methods. We introduce an \Q@if-elseif-else@-expression where the branches invoke abstract methods. The guards check whether the list has only one element or whether the current element is larger than or equal to the maximum of the rest of the list. The abstract method \Q@accessHead@ returns the current element, and the abstract method \Q@maxTail@ returns the maximum in the remaining list. So, we recursively search the list for the largest element by comparing the maximum element of the list tail with the current element until we reach the end of the list.

\begin{lstlisting}
trait MaxETrait2 {
  @Pre: list.size() > 0
  @Post: list.contains(result) &
    (forall Num n: list.contains(n) ==> result >= n)
  Num maxElement(List list) =
    if (list.size() == 1) {accessHead(list)}
    elseif (accessHead(list) >= maxTail(list)) 
      {accessHead(list)}
    else {maxTail(list)}

  @Pre: list.size() > 0
  @Post: result == list.element()
  abstract Num accessHead(List list);

  @Pre: list.size() > 1
  @Post: list.tail().contains(result) &
    (forall Num n: list.tail().contains(n) ==> result >= n)
  abstract Num maxTail(List list);
}
\end{lstlisting}



The correct implementation of the method \Q@maxElement@ can be guaranteed under the assumptions that all introduced abstract methods are correctly implemented.
Similar to PhV, a program verifier conducts a proof of method \Q@maxElement@ and uses the introduced specifications of the methods \Q@accessHead@ and \Q@maxTail@. When the proof succeeds, we know that the first method is correctly implemented. In our incremental CbCTrait process, we verify each method implementation directly after construction; and so we are able to reuse each implemented method in the following steps (e.g., by calling the method in the body of other methods).


We now compose the developed traits to complete the first refinement step. To perform the composition \Q@MaxETrait1@ + \Q@MaxETrait2@, we check that the specification of the method \Q@maxElement@ fulfills the specification of the abstract method in the first trait (cf. Liskov substitution principle~\cite{liskov}). In this case, this means checking that:\\
\Q@MaxETrait1.maxElement(..).pre ==> MaxETrait2.maxElement(..).pre@
as well as: 
\Q@MaxETrait2.maxElement(..).post ==> MaxETrait1.maxElement(..).post@.\\
When the composition of two verified traits is successful, the result is also a verified trait.
Note that the composed trait does not need to be verified directly by a program verifier in TraitCbC because it is correct by construction.
In this example, the specifications are the same, thus checking for a successful composition is trivial, but this is not generally the case. 
In particular, the logic needs to take into account ill-founded specifications and recursion in the specification.
We discuss more about the difficulties of handling those cases in the Appendix~\ref{sec:circularity}.

The methods \Q@accessHead@ and \Q@maxTail@ are implemented in the next two refinement steps in traits \Q@MaxETrait3@ and \Q@MaxETrait4@\footnote{The methods could also be implemented in one trait.}. As we implement a recursive method, the method \Q@maxTail@ calls the \Q@maxElement@ method, thus \Q@maxElement@ is introduced as an abstract method in this trait. 
We have to verify that the method \Q@accessHead@ satisfies its specification using a program verifier.
Similarly, we have to verify the correctness of the method \Q@maxTail@.

\begin{lstlisting}
trait MaxETrait3 {
  @Pre: list.size() > 0
  @Post: result == list.element()
  Num accessHead(List list) = list.element()
}
\end{lstlisting}

\begin{lstlisting}
trait MaxETrait4 {
  @Pre: list.size() > 1
  @Post: list.tail().contains(result) &
    (forall Num n: list.tail().contains(n) ==> result >= n)
  Num maxTail(List list) = maxElement(list.tail())

  @Pre: list.size() > 0
  @Post: list.contains(result) &
    (forall Num n: list.contains(n) ==> result >= n)
  abstract Num maxElement(List list);
}
\end{lstlisting}

As before, all traits are composed, and it is checked that the specifications of the concrete methods fulfill the specifications of the abstract ones.
As we have no contradicting specifications for the same methods, the composition is well-formed. 
The final program \Q@MaxE@ is as follows. 

\begin{lstlisting}
class MaxE = MaxETrait1 + MaxETrait2 + MaxETrait3 + MaxETrait4 
\end{lstlisting}

%
%


\paragraph{Advantages of TraitCbC}
As shown in the example, TraitCbC enables the CbC programming style without the need of external refinement rules.
In classical CbC, when designing a unit of code, the programmer has to proceed with atomic steps of a predefined granularity.
In contrast, in TraitCbC the programmer is free to divide a unit of code in any granularity, by including
as many auxiliary methods as needed to bring the verification to an appropriate granularity. TraitCbC helps to construct code in fine-grained steps which are more amenable for verification than single more complex methods.
If the programmer chooses to not include any auxiliary methods at all, this is 
essentially the same as the traditional post-hoc verification style.
In the example above, we could implement the method \Q@maxElement@ in one step without the intermediate step that introduces the two abstract methods \Q@accessHead@ and \Q@maxTail@.

Additionally, the already proven auxiliary methods in traits can be reused. 
For example, if we want to implement a \Q@minElement@ method,
we could reuse already implemented traits to reduce the programming and verification effort.
The method \Q@minElement@ is implemented in the following in trait \Q@MinE@ with one abstract method. The specification of the method \Q@accessHead@ is the same as for the method \Q@accessHead@ above, so \Q@MaxETrait3@ can be reused.
In this example, we show the flexible granularity of TraitCbC by directly implementing the else branch, instead of introducing an auxiliary method as for \Q@maxElement@.

\begin{lstlisting}
trait MinE {
  @Pre: list.size() > 0
  @Post: list.contains(result) &
    (forall Num n: list.contains(n) ==> result <= n)
  Num minElement(List list) =
    if (list.size() == 1) {accessHead(list)}
    elseif (accessHead(list) <= minElement(list.tail()))
      {accessHead(list)}
    else {minElement(list.tail())}
  
  @Pre: list.size() > 0
  @Post: result == list.element()
  abstract Num accessHead(List list);
}
\end{lstlisting}

The correctness of \Q@minElement@ is verified with the specifications of the method \Q@accessHead@. 
By composing \Q@MinE@ with \Q@MaxETrait3@, we get a correct implementation of \Q@minElement@.
Note how this verification process supports abstraction:
as long as the contracts are compatible,
methods can be implemented in different styles by different programmers to best meet non-functional requirements while preserving the specified observable behavior~\cite{terBeek2018xbc}.
A completely different implementation of \Q@maxElement@ can be used if it fulfills the specification of the abstract method \Q@maxElement@ in trait \Q@MaxETrait1@. This decoupling of specification and corresponding satisfying implementations facilitates an incremental development process where a specified code base is extended with suitable implementations~\cite{damiani2014verifying}.

\section{Object-Oriented Trait-Based Language}
In this section, we formally introduce the syntax, type system, and flattening semantics of a minimal core calculus for TraitCbC.
We keep this calculus for TraitCbC parametric in the specification logic so that it can be used with a suitable program verifier and associated logic.
The presented rules to compose traits are conventional. The focus of our work is to enable a CbC approach using traits that programmers can easily adopt. Therefore, we present the calculus to prove soundness of TraitCbC, but focus on the presentation of the advantages of incremental trait-based programming in this paper. Indeed, languages with traits and with a suitable specification language intrinsically enable incremental program construction. For the sake of completeness, reduction rules of TraitCbC are presented in Appendix~\ref{sec:reduction}.

\subsection{Syntax}
The concrete syntax of our core calculus for TraitCbC is shown in Fig.~\ref{f:syntax}, where non-terminals ending with `s' are implicitly defined as a sequence of non-terminals, i.e., $vs ::= v_1 \dots v_n$. We use the metavariables $t$ for trait names, $C$ for class names and $m$ for method names.
A program consists of trait and class definitions. Each definition has a name and a trait expression $\mathit{E}$. The trait expression can be a $\mathit{Body}$, a trait name, a composition of two trait expressions $\mathit{E}$, or a trait expression $\mathit{E}$ where a method is made abstract, written as $E[\mathtt{makeAbstract}\ m]$. A $\mathit{Body}$ has a flag $\mathtt{interface}$ to define an interface, a set of implemented interfaces $\Cs$ and a list of methods $\Ms$. Methods have a method header $\mathit{MH}$ consisting of a specification $S$,  the return type, a method name, and a list of parameters. Methods have an optional method body. In the method body, we have standard expressions, such as variable references, method calls, and object initializations. For simplicity, we exclude updatable state. Field declarations are emulated by method declarations, and field accesses are emulated by method calls.

The specification $\mathit{S}$ in each method header is used to verify that methods are correctly implemented. The specification is written in some logic. In our examples, we will use first-order logic (cf. the example in Section~\ref{sec:example}).
A well-formed program respects the following conditions:

Every $\mathit{Name}$ in $\mathit{Ds}$ must be unique so that $\mathit{Ds}$ can be seen as a map
from names to trait expressions.
Trait expressions $E$ can refer to trait names $t$.
A well-formed $\Ds$ does not have any circular trait definitions like $t = t$
or $t_1 =t_2$ and $t_2 = t_1$.
In a $\Body$, all names of implemented interfaces must be unique and all method names must be unique, so that $\Body$ is a map
from method names to method definitions.
In a method header, parameters must have unique names, and no explicit parameter
can be called $\mathtt{this}$.

%
%
%
%
 
\begin{figure}[t]
	\centering
	\setlength{\fboxsep}{3pt}
	\setlength{\fboxrule}{1pt}
	\fbox{
		\begin{minipage}{0.97\linewidth}
			\begin{array}[t]{lcl}
				\mathit{Prog} &::= &\Ds\ e\\
				D &::= &\TD\ |\ \CD\\
				\Name &::= &t\ |\ C\\
				\TD &::=  &t = E\\
				\CD &::= &C = E\\
				E &::= &\Body\ |\ t\ |\ E + E\ |\ E[\mathtt{makeAbstract}\ m]\\
				\Body &::= &\{\mathtt{interface}?\ [\Cs]\ \Ms \}\\
				M &::= &\MH\ e?;\\
				\MH &::= &S\ \methodKW\ C\ m(C_1\ x_1 \dots C_n\ x_n)\\
				e &::= &x\ |\ e.m(\es)\ |\ \newKW\ C(\es)\\
				\E_v &::= &[].m(es)  \ |\ v.m(vs\ []\ es)\ |\ \newKW\ C(\vs\ []\ \es)\\
				v &::= &\newKW\ C(\vs)\\
				\Gamma &::= &x_1:C_1 \dots x_n:C_n\\
				S &::= & \dots e.g.\ \mathtt{Pre}:P\ \mathtt{Post}:P\\
				P &::= & \dots e.g.\ \text{First order logic}
			\end{array}
	\end{minipage}}
	\caption{Syntax of the trait system}
	\label{f:syntax}
\end{figure}


\subsection{Typing Rules}

In our type system, we have a typing context $\Gamma ::= x_1:C_1\dots x_n:C_n$ which assigns types $C_i$ to variables $x_i$.
We define typing rules for our three kinds of expressions: $x$, method calls, and object initialization. We combine typing and verification in our type checking $	\Gamma \vdash e : C \dashv P_0 \models P_1$. This judgment can be read as: under typing context $\Gamma$, the expression $e$ has type $C$, where under the knowledge $P_0$ we need to prove $P_1$. The knowledge $P_0$ is our collected information that we use to prove a method correct.
That means, in our typing rules, we collect the knowledge about the parameters and expressions in a method body to verify that this method body fulfills the specification defined in the method header. The verification obligation $P_1$ should follow from the knowledge $P_0$.

We check if methods are well-typed with judgments of form
$\Ds; \Name \vdash M : \OK$. This judgment can be read as: in the definition table, the method $M$ defined under the definition \Name\ is correct.
The typing rules of Fig.~\ref{f:typing} are explained in Appendix~\ref{sec:typing} in detail. 
The first four rules type different expressions and collect the information of these expressions to prove with rule \textsc{MOK} that a method fulfills its specification. In the rule \textsc{MOK} with keyword \textbf{verify}, we call a verifier to prove each method once. Abstract methods (\textsc{AbsOK}) are always correct. Rule \textsc{BodyOK} ensures that all methods in a body are correctly typed.

\begin{figure}[t]
	\centering
	\begin{scriptsize}
		\begin{mathpar}
			\inferrule
			{
			}{
				\Gamma \vdash x: \Gamma(x) \dashv \result: \Gamma(x)\ \&\ \result = x \models \mathit{true}}
			\quad (\textsc{x})
			
			\inferrule
			{
				S\ \methodKW\ C\ m (C_1\ x_1 \dots C_n\ x_n) \_;\ \in\ \mathit{methods(C_0)}\\
				\Gamma \vdash e_0 : C_0 \dashv P_0 \models P'_0\
				\dots\
				\Gamma \vdash e_n : C_n \dashv P_n \models P'_n\\
				x'_0 \dots x'_n\ \mathit{fresh}\\
				S' = S[\this:=x'_0,\ x_1 := x'_1,\ \dots,\ x_n:= x'_n]\\
				P = (\result:C\ \&\ P_0[\result:=x'_0]\ \&\ \dots\ \&\ P_n[\result:=x'_n]\\\\ \&\ (\Pre(S') \implies \Post(S')))
			}{
				\Gamma \vdash e_0.m(e_1\dots e_n) : C \dashv P \models P'_0\ \&\ \dots\ \&\ P'_n\ \&\ \Pre(S')}
			\quad (\textsc{Method})
			
			\inferrule
			{
				\Gamma \vdash e_1 : C_1 \dashv P_1 \models P'_1\
				\dots\
				\Gamma \vdash e_n : C_n \dashv P_n \models P'_n\\
				\mathit{getters}(C) = S_1\ \methodKW\ C_1\ x_1();\ \dots\ S_n\ \methodKW\ C_n\ x_n();\\
				x'_1 \dots x'_n\ \mathit{fresh}\\
				S'_i = S_i[\this:=\result]\\
				P''_i = (P_i[\result:=x'_i]\ \&\ (\Pre(S'_i) \implies \result.x_i()=x'_i))\\
				P = (\result:C\ \&\ P''_1\ \&\ \dots\ \&\ P''_n)
			}{
				\Gamma \vdash \newKW\ C(e_1 \dots e_n) : C \dashv P \models P'_1\ \&\ \dots\ \&\ P'_n\ \&\ \Pre(S'_1)\ \&\ \dots \&\ \Pre(S'_n)}
			\quad (\textsc{New})
			
			\inferrule
			{
				\Gamma \vdash e : C' \dashv P \models P'\\
				C'\ instanceof\ C
			}{
				\Gamma \vdash e: C \dashv P \models P'}
			\quad (\textsc{sub})
			
			\inferrule
			{
				\Gamma = \this:\Name,\ x_1:C_1,\ \dots,\ x_n:C_n\\
				\Gamma \vdash e : C \dashv P \models P'\\\\
				\textbf{verify}\ Ds \vdash (\Gamma\ \&\ \Pre(S)\ \&\ P) \models (P'\ \&\ \Post(S))
			}{
				\Ds;\ \Name \vdash S\ \methodKW\ C\ m (C_1\ x_1 \dots C_n\ x_n)\ e;\ : \OK}
			\quad (\textsc{MOK})
			
			\inferrule
			{
			}{
				\Ds;\ \Name \vdash S\ \methodKW\ C\ m(C_1\ x_1 \dots C_n\ x_n);\ : \OK}
			\quad (\textsc{AbsMOK})
			
			\inferrule
			{
				\Body = \{\interface?\ [\Cs]\ M_1 \dots M_n\}\\\\
				\Ds; \Name \vdash M_1 : \OK \dots \Ds; \Name \vdash M_n : \OK
			}{
				\Ds; \Name \vdash \Body\ : \OK}
			\quad (\textsc{BodyTyped})	
		\end{mathpar}
	\end{scriptsize}
	\caption{Expression typing rules}
	\label{f:typing}
\end{figure}

\subsection{Flattening Semantics}
When we implement methods in several traits, we have to check that these traits are compatible when they are composed. This process to derive a complete class from a set of traits is called flattening. We follow the traditional flattening semantics~\cite{ducasse2006traits}.
A class that is defined by composing several traits is obtained by flattening rules. All methods are direct members of the class~\cite{ducasse2006traits}. Overall, our flattening process works as a big step reduction arrow, where we reduce a trait expression into a well-typed and verified body.




To introduce our flattening rules in Fig~\ref{f:flattening}, we first define the helper functions.
The function $\mathit{allMeth}$ collects all method headers with the same name as $m$ in all input bodies (Definition~\ref{def:allmeth}).
When two $\mathit{Body}$s are composed (Definition~\ref{def:bodyplus}), the implemented interfaces are united and the methods are composed. 
The composition of methods (Definition~\ref{def:msplus}) collects methods that are only defined in one of the input sets. If a method is in both sets, it is composed (Definition~\ref{def:mplus}). Here, we distinguish four cases. If one method is abstract and the other is concrete, we have to show that the precondition of the abstract method implies the precondition of the concrete method. Additionally, the postcondition of the concrete one has to imply the postcondition of the abstract one. This is similar to Liskov's substitution principle~\cite{liskov}.
The second case is the symmetric variant of the first case. In the third and fourth case, two abstract methods are composed. Here, the specification of one abstract method has to imply the specification of the other abstract method such that an implementation can still satisfy all specifications of abstract methods. If both method are concrete, the composition is correctly left undefined. This composition error can be resolved by making one method $m$ abstract in the $\Body$, as defined in Definition~\ref{def:bodyabs}. The resulting $\Body$ is similar with the difference that the implementation of the method $m$ is omitted.
The flattening rules in Fig.~\ref{f:flattening} are explained in the following in detail. In these rules, a set of traits is flattened to a declaration containing all methods. If abstract and concrete methods with the same name are composed, Definitions \ref{def:bodyplus}-\ref{def:mplus} are used to guarantee correctness of the composition.

\begin{definition}[All Methods]\label{def:allmeth}
	$\mathit{allMeth}(m,\ \mathit{Bodys}) =$\\
	${}_{}\quad \{\MH;\ |\ \Body\ \in\ \mathit{Bodys},\ \Body(m)=\MH;\}$
\end{definition}
	\begin{definition}[Body Composition]\label{def:bodyplus} $\Body_1 + \Body_2 = \Body$\\
		${}_{} \{ \interface?\ [\Cs_1]\ \Ms_1 \} + \{ \interface?\ [\Cs_1]\ \Ms_1 \} =$\\ 
		${}_{}\{ \interface?\ [\Cs_1\ \cup \Cs_2]\ \Ms_1+\Ms_2 \}$
	\end{definition}
	\begin{definition}[Methods Composition]\label{def:msplus} $\Ms_1+\Ms_2=\Ms$\\
		${}_{}\bullet (M\ \Ms_1) + \Ms_2 = M\ (\Ms_1+\Ms_2)$\\
		${}_{}\quad  \mathit{if\ methName(M)}\notin\mathit{dom (\Ms_2)}$\\
		${}_{}\bullet (M_1\ \Ms_1) + (M_2\ \Ms_2) = M_1+M_2\ (\Ms_1+\Ms_2)$\\
		${}_{}\quad \mathit{if\ methName(M_1)= methName(M_2)}$\\
		${}_{}\bullet \emptyset + \Ms = \Ms$
	\end{definition}
	\begin{definition}[Method Composition]\label{def:mplus} $M_1+M_2=M $\\
		${}_{}\bullet S\ \methodKW\ C\ m(C_1\ x_1 \dots C_n\ x_n)\ e;\ +\ S'\ \methodKW\ C\ m(C_1\ \_\dots C_n\ \_);$\\
		${}_{}\quad =\ S\ \methodKW\ C\ m(C_1\ x_1 \dots C_n\ x_n)\ e;$\\
		${}_{}\quad \mathit{if}\  \Pre(S')\ \mathit{implies}\ \Pre(S)\ and\ \Post(S)\ \mathit{implies}\ \Post(S')$\\
		${}_{}\bullet \MH_1;\ +\ \MH_2\ e;\quad  =\quad  \MH_2\ e;\ +\ \MH_1;$\\
		${}_{}\bullet S\ \methodKW\ C\ m(C_1\ x_1\dots C_n\ x_n);\ +\ S'\ \methodKW\ C\ m(C_1\ \_\dots C_n\ \_);$\\
 		${}_{}\quad =\ S\ \methodKW\ C\ m(C_1\ x_1\dots C_n\ x_n);$\\
		${}_{}\quad \mathit{if}\  \Pre(S')\ \mathit{implies}\ \Pre(S)\ \mathit{and\ Post(S)\ implies\ Post(S')}$\\
		${}_{}\bullet S\ \methodKW\ C\ m(C_1\ x_1\dots C_n\ x_n);\ +\ S'\ \methodKW\ C\ m(C_1\ \_\dots C_n\ \_);$\\
		${}_{}\quad =\ S'\ \methodKW\ C\ m(C_1\ x_1\dots C_n\ x_n);$\\
		${}_{}\quad \mathit{if\  (Pre(S)\ implies\ Pre(S')\ and\ Post(S')\ implies\ Post(S))}$\\
		${}_{}\quad \mathit{and\ not\  (Pre(S')\ implies\ Pre(S)\ and\ Post(S)\ implies \Post(S'))}$
	\end{definition}
	\begin{definition}[Body Abstraction]\label{def:bodyabs} $\Body[\mathtt{makeAbstract}\ m]$\\
	${}_{}\quad \{ [\Cs]\ \Ms_1\ S\ \mathtt{method}\ C\ m(\mathit{Cxs}) \_;\  \Ms_2\}[\mathtt{makeAbstract}\ m] $\\ 
	${}_{}\quad = \{ [\Cs]\ \Ms_1\ S\ \mathtt{method}\ C\ m(\mathit{Cxs});\  \Ms_2\}$
	\end{definition}

\begin{figure}[!t]
	\centering
	\begin{scriptsize}
	\begin{mathpar}
		\inferrule
		{
			D'_1 \dots D'_n \vdash D_1 \Downarrow D'_1\\ \dots\\ D'_1 \dots D'_n \vdash D_n \Downarrow D'_n
		}{
			D_1 \dots D_n \Downarrow D'_1 \dots D'_n}
		\quad (\textsc{FlatTop})
		
		\inferrule
		{
			\Ds;\ \Name \vdash E \Downarrow \Body\\
			\mathtt{if}\ \Name\ \mathtt{of\ form}\ C\ \mathtt{then}\ \mathit{abs(Body)} = S\ T\ x_1(); \dots S\ T\ x_n();
		}{
			\Ds \vdash \Name = E \Downarrow \Name = \Body}
		\quad (\textsc{DFlat})
		
		\inferrule
		{
			\Body = \{\interface?\ [\Cs]\ M_1 \dots M_n\}\\\\
			\Body' = \{\interface?\ [\Cs]\ M_1 \dots M_n\ \Ms\}\\
			\Ms = \{\Sigma \mathit{allMeth}(\Ds,\ \Cs,\ m)\ |\  m \in \mathit{dom(Cs)}\ \mathit{and}\ m \notin \mathit{dom(Body)} \}\\
			\Ds;\ \Name \vdash \Body' : \OK
		}{
			\Ds;\ \Name \vdash \Body \Downarrow \Body'}
		 \quad (\textsc{BFlat})
		
		\inferrule
		{
		}{
			\Ds;\ \Name \vdash t \Downarrow \Ds(t)}
		\quad (\textsc{tFlat})
		
		\inferrule
		{
			\Ds;\ \Name \vdash E_1 \Downarrow \Body_1\\
			\Ds;\ \Name \vdash E_2 \Downarrow \Body_2\\
		}{
			\Ds;\ \Name \vdash E_1 + E_2 \Downarrow \Body_1 + \Body_2}
		\quad (\textsc{+Flat})
		
		\inferrule
		{
			\Ds;\ \Name \vdash E \Downarrow \Body\\
			\Body = \{[\Cs]\ \overline{M}_1\ S\ \mathtt{method}\ C\ m(C_1\ x_1 \dots C_n\ x_n) \_;\ \overline{M}_2\}\\
			\Body' = \{[\Cs]\ \overline{M}_1\ S\ \mathtt{method}\ C\ m(C_1\ x_1 \dots C_n\ x_n);\ \overline{M}_2\}
		}{
			\Ds;\ \Name \vdash E[\mathtt{makeAbstract}\ m] \Downarrow \Body'}
		\quad (\textsc{AbsFlat})
	\end{mathpar}
	\end{scriptsize}
	\caption{Flattening rules}
	\label{f:flattening}
\end{figure}

	\textsc{FlatTop}. The first rule flattens a set of declarations $D_1\dots D_n$ to a set $D'_1 \dots D'_n$.	We express this rule in a non-computational way:
	we assume to know the resulting $D'_1 \dots D'_n$, and we use them as a guide to
	compute them.
	Note that if there is a resulting $D'_1 \dots D'_n$ then it is unique;
	flattening is a deterministic process and $D'_1 \dots D'_n$ are used only
	to type check the results. They are not used to compute the shape of
	the flattened code.
	
	Non computational rules like this are common with nominal type
	systems~\cite{igarashi2001featherweight} where the type signatures of all classes and methods can
	be extracted before the method bodies are verified.
	
	\textsc{DFlat}. This rule flattens an individual definition by flattening the trait expression.
	When the flattening produces a class definition, we also check that the body
	denotes an instantiable class; a class whose only abstract methods are
	valid getters. The function $\mathit{abs(Body)}$ returns the abstract methods.
	
	\textsc{BFlat} It may look surprising that the $\Body$ does not flatten to itself. This represents what happens in most programming languages, where
	implementing an interface implicitly imports the abstract signature
	for all the methods of that interface.
	In the context of verification also the specification of such interface methods is imported.
	In concrete, $\Body'$ is like $\Body$, but we add $\Ms$ by collecting all the methods
	of the interfaces that are not already present in the $\Body$.
	
	Moreover, we check that all the methods defined in the class respect
	the typing and the specification defined in the interfaces:
	if a class has $S\ $\Q@method Foo@ \Q@foo();@ or 
	$S\ $\Q@method Foo foo() e;@
	and there is a $S'\ $\Q@method Foo foo();@ in the interface,
	then $S$ must respect the specification $S'$.
	The system then checks that the $\Body$ is well-typed and verified by calling
	$\Ds;\ \Name \vdash M_i : \OK$
	
	
	\textsc{TFlat}. A trait $t$ is flattened to its declaration $\Ds(t)$.
	
	\textsc{+Flat}. The composition of two expression $E_1$ and $E_2$, where both expressions are first reduced to $\Body_1$ and $\Body_2$, results in the composition of these bodies as defined in Definition~\ref{def:bodyplus}.
	
	\textsc{AbsFlat}. An expression $E$ where one method $m$ is made abstract flattens to a $\Body'$. We know that $E$ flattens to $\Body$. The only difference between $\Body$ and $\Body'$ is that the one method $m$ is abstract in $\Body'$. In $\Body$, the method can be abstract or concrete.

\subsection{Soundness of the Trait-based CbC Process}

In this section, we formulate our main result of the TraitCbC process. We prove soundness of the flattening process with a parametric logic. The proofs of the lemmas and theorems are in Appendix~\ref{sec:proofs}.
We claim that if you have a language without code reuse and with sound and modular PhV verification
then the language supports CbC simply by adding traits to the language. That is, traits intrinsically enable a CbC program construction process.


To prove soundness of the refinement process of TraitCbC (Theorem~\ref{theorem:cbcsound}: Sound CbC Process) as exemplified in Section~\ref{sec:example}, we have to show that the flattening process is correct (Theorem~\ref{theorem:soundness}: General Soundness). 
In turn, to prove General Soundness, we need two lemmas which state that the composition of traits is correct (Lemma~\ref{lemma:flattening}) and that a trait after the $\mathtt{makeAbstract}$ operation is still correct (Lemma~\ref{lemma:makeAbstract}).

In Lemma~\ref{lemma:flattening}, we have well-typed definitions $\Ds$, and two well-typed and verified traits in $\Ds$, and the resulting trait/class is also well-typed and verified.

\begin{lemma}[Composition correct]\label{lemma:flattening}\quad\\
If $\Ds(t1) = \Body_1$,
$\Ds(t2) = \Body_2$, 
$\Ds(\Name) = \Body$,
$\Ds; t_1 \vdash \Body_1 : \OK$,
${}_{}\quad$ $\Ds; t_2 \vdash \Body_2 : \OK$,
and $\Body_1 + \Body_2=\Body$,\\*
then $\Ds; \Name \vdash \Body : \OK$
\end{lemma}

Lemma~\ref{lemma:makeAbstract} shows that if we have a well-typed and verified trait, the
operation $\mathtt{makeAbstract}$ results in a trait/class that is also well-typed and verified.

\begin{lemma}[MakeAbstract correct]\label{lemma:makeAbstract}\quad\\
If $\Ds(t) = \Body$,
$\Ds(\Name) = \Body'$,
$\Ds; t \vdash \Body : \OK$,\\* 
${}_{}\quad$ and $\Body[\mathtt{makeAbstract}\ m] = \Body'$,\\*
then $\Ds; \Name$ $\vdash \Body' : \OK$
\end{lemma}

With these Lemmas, we can prove Theorem~\ref{theorem:soundness}. Given a sound and modular verification language, then all programs that flatten are well-typed and verified.
In a modular verification language, a method can be fully verified using only the information contained in the method declaration and the specification of any used method.
Moreover, our parametric logic must support at least a commutative and associative \textit{and} (but of course other ways to merge knowledge could work too) and a transitive \textit{implication} (but of course other forms of logical consequence could work too).

\begin{theorem}[General Soundness]\label{theorem:soundness}\quad\\
For all programs $\Ds$ where $\Ds$ flattens to $\Ds'$, and $\Ds'$ is well-typed;\\*
that is, forall $\Name=\Body \in \Ds'$,
we have $\Ds';\ \Name \vdash \Body : \OK$.
\end{theorem}

We now show that the TraitCbC process is sound. Theorem~\ref{theorem:cbcsound} states that starting with one abstract method and a set of verified traits, the composed program is also verified.

\begin{theorem}[Sound CbC Process]\label{theorem:cbcsound}\quad\\
	Starting from a fully abstract specification $t_0$, and some refinement steps $t_1 \dots t_n$, we can write $C = t_0 + \dots + t_n$ as our whole CbC refinement process; where $t_0 + t_1$ is the application of the first refinement step. 
	If we use CbC to construct programs, we can start from verified atomic units and get a verified result. Formally, if $t_0 = \{\MH\}\
	t_1 = \{\Ms_1\}\
	\dots\
	t_n = \{\Ms_n\}$ are well-typed, and\\
	$\begin{array}{lcl}
		t_0 = \{\MH\}&&t_0 = \{\MH\}\\
		t_1 = \{\Ms_1\}\
		\dots\
		t_n = \{\Ms_n\} &\quad\Downarrow \quad\quad & t_1 = \{\Ms_1\}\
		\dots\
		t_n = \{\Ms_n\}\\
		C = t_0 + \dots + t_n && C = \Body\\
	\end{array}$\\
	then $C=\Body$ is well-typed.
\end{theorem}
\begin{proof}
	This is a special case of Theorem~\ref{theorem:soundness}.
\end{proof}

Theorem~\ref{theorem:cbcsound} shows clearly that trait composition intrinsically enables a CbC refinement process: 
A object-oriented programming language with traits
and a corresponding specification language supports an incremental CbC approach.

\section{Trait-based Correctness-by-Construction in Comparison to Classical CbC}
\label{sec:discussion}
In this section, we discuss the benefits of TraitCbC in comparison to classical CbC. To do this, we describe classical CbC first.

Classical correctness-by-construction (CbC)~\cite{dijkstra-book,kourie2012correctness,morgan-book} is an incremental approach to construct programs. 
CbC uses a \textit{Hoare triple} specification $\hoare{P}{S}{Q}$ stating that if the precondition $\mathtt{P}$ holds, and the statement $\mathtt{S}$ is executed, then the statement terminates and postcondition $\mathtt{Q}$ holds.
The CbC refinement process starts with a Hoare triple where the statement $\mathtt{S}$ is abstract. This abstract statement can be seen as a hole in the program that needs to be filled. 
With a set of refinement rules, an abstract statement is replaced by more concrete statements (i.e., statements in the guarded command language~\cite{dijkstra-book} that can contain further abstract statements). The process stops, when all abstract statements are refined to concrete statements so that no holes remain in the program. 
As each refinement rule is sound and each correct application of a refinement rule guarantees to satisfy the starting Hoare triple, the resulting program is correct-by-construction~\cite{kourie2012correctness}.
The CbC process is strictly tied to a set of predefined refinement rules. A programmer cannot deviate from this concept.
To apply a refinement rule, it has to be checked that conditions of the rule application are satisfied. This is done by pen-and-paper or with specialized tools~\cite{runge2019tool}.

\begin{table*}[tb]
	\restoretx
	\centering
	\begin{scriptsize}
	\begin{tabularx}{\linewidth}{nXX}
		\toprule
		 & Classic CbC	& TraitCbC\\ 
		\hline
		Language & Additional rules for a programming language.	& Programming language with traits. Needs specification language.\\
		\hline 
		Tool \newline support & Pen and paper. Some specialized tools available. & Relies on prevalent PhV verification tools.\\ 
		\hline
		Construc- tion Rules & Specific refinement rules. & Refinement by composition of traits.\\ 
		\hline 
		Debugging & Guarantees the correctness of each refinement step. Only refinements without abstract statement are directly verified. & Guarantees the correctness of each refinement step. Each method is specified such that each refinement can directly be verified.\\
		\hline
		Proof\newline complexity & Many, but small proofs. & Any granularity of proofs.\\
		\hline
		Reuse & 
Refinement steps cannot be reused; only fully implemented methods can.
 & Each verified method in a trait can be reused.\\
 		\hline
 Applications & Focuses on small but correctness-critical algorithms.
 & As TraitCbC is based on PhV, it can be used in areas of PhV. Additionally, traits are beneficial for incremental development approaches and development of software product lines.\\
		\bottomrule
	\end{tabularx}
	\end{scriptsize}
	\caption{Comparison of TraitCbC with classical CbC}
	\label{tab:compare}
\end{table*}

In Table~\ref{tab:compare}, we compare TraitCbC and classical CbC: 

\noindent\emph{Language} The classical CbC approach is external to a programming language. It needs the definition of refinement rules. TraitCbC is usable with languages that have traits, a specification language, and a corresponding verification framework.
In this work, we focus on object-orientation, but the general TraitCbC programming guideline presented in this paper is also suitable for functional programming environments using abstract and concrete functions with specifications instead of traits and methods.

\noindent\emph{Tool Support}
To use one of the approaches, tool support is desired. For classical CbC, mostly pen and paper is used. There are a few specialized tools such as CorC~\cite{runge2019tool}, tool support for ArcAngel~\cite{oliveira2003arcangel}, and SOCOS~\cite{back2009invariant,back2007testing}. These tools force a certain programming procedure on the user. This procedure can be in conflict with their preferred programming style.
For TraitCbC, tools for post-hoc verification can be reused. There are tools for many languages such as Java~\cite{ahrendt2016deductive}, C~\cite{cohen2009vcc}, C\#~\cite{Barnett:2011:SVS:1953122.1953145,barnett2004spec}.
Other languages are integrated with their verifier from the start, e.g., Spec\#~\cite{barnett2004spec} and Dafny~\cite{leino2010dafny}.
 TraitCbC as presented in this paper is a core calculus, designed to show the feasibility of the concept. 
We believe that scaling up TraitCbC to a complete programming language reusing existing verification techniques would be feasible
and would result in a similarly expressive verification process, but supporting more flexible program composition.
In Section~\ref{sec:impl}, we show how a prototype can be constructed by using the KeY verifier~\cite{ahrendt2016deductive}.

\noindent\emph{Construction Rules}
To construct a program, classical CbC has a strict concept of refinement rules. A programmer cannot deviate from the granularity of the rules. In contrast, PhV does not give a mandatory guideline how to construct programs. TraitCbC is a bridge between both extremes. Programs can be constructed stepwise as with classical CbC, but if desired, any number of refinement steps can be condensed up to PhV based programming.

\noindent\emph{Debugging}
If errors occur in the development process, TraitCbC gives early and detailed information. By specifying the method under development and any abstract method that is called by this method, we can directly verify the correctness of the method under development.
We assume that the introduced abstract methods will be correctly implemented in further refinement steps. With each step, the programmer gets closer to the solution until finally all abstract methods are implemented. Classical CbC relies on the same process, but here the abstract statements (similar to our abstract methods) are not explicitly specified by the programmer. Additional specifications in classical CbC are introduced only with some rules such as an intermediate condition in the composition rule. Then, these specifications are propagated through the program to be constructed. When arriving at a leaf in the refinement process, the correctness of the statement can be guaranteed. The problem in classical CbC is that all refinement steps where abstract statements occur cannot be verified directly. In the worst case, a wrong specification is found only after a few refinement steps.

\noindent\emph{Proof Complexity}
TraitCbC can have the same granularity and also the same proof effort as classical CbC, since each method implementation can correspond to just one refinement step. The advantage of TraitCbC is that programmers can freely implement a method body. They must not stick to the same granularity as in the classical CbC refinement rules. As in PhV, they can implement a complete method in one step. 
The programmer can balance proof complexity against verifier calls.

\noindent\emph{Reuse}
If we want to reuse developed methods or refinement steps, the approaches differ. In classical CbC, no refinement steps can be reused. A fully refined method can be reused in both approaches. For TraitCbC, we can easily reuse even very small units of code, since they are represented as methods in the traits.

\noindent\emph{Applications}
The classical CbC approach does not scale well to development procedures for complete software system. Rather, individual algorithms can be developed with CbC~\cite{watson2016correctness}. As soon as we scale TraitCbC to real languages, we have the same application scenarios as PhV. As argued by Damiani et al.~\cite{damiani2014verifying} traits enable an incremental process of specifying and verifying software. Bettini et al.~\cite{bettini2010implementing} proposed to use traits for software product line development and highlighted the benefits of fine-grained reuse mechanisms. Here, TraitCbC's guideline is suitable for constructing new product lines step by step from the beginning.

\noindent\emph{Summary}
In summary, TraitCbC bridges the gap between PhV and CbC. It enables a CbC process for trait-based languages without introducing refinement rules.
The concrete realization of specifying and verifying methods is similar to PhV, but additionally to PhV, TraitCbC provides an incremental development process. This development process combined with the flexibility of traits allows correct methods to be developed in small and reusable steps. Moreover, we have introduced a core calculus and proved that the construction and composition of trait-based programs is correct.

\section{Proof-of-Concept Implementation}
\label{sec:impl}

In this section, we describe the implementation, which instantiates TraitCbC in Java with JML~\cite{leavens1998jml} as specification language and KeY~\cite{ahrendt2016deductive} as verifier for Java code. Our trait implementation is based on interfaces with default implementation. Our open source tool is implemented in Java and integrated as plug-in in the Eclipse IDE.\footnote{Tool and evaluation at \url{https://github.com/TUBS-ISF/CorC/tree/TraitCbC}}
Besides this prototype, other languages with a suitable verifier, such as Dafny~\cite{leino2010dafny} and OpenJML~\cite{cok2011openjml}, can also be used to implement TraitCbC.

\hl{In Listing~\ref{code:exampleJava}, we show the concrete syntax. Each method in a trait is specified with JML with the keywords \Q@requires@ and \Q@ensures@ for the pre- and postcondition. 
To verify the correctness of programs, we need two steps. First, we verify the correctness of a method implemented in a trait w.r.t. its specification. Second, for trait composition, our implementation checks the correct composition for all methods (cf. Definition~\ref{def:bodyplus}). It is verified that the specification of a concrete method satisfies the specification of the abstract one with the same signature (cf. Definition~\ref{def:mplus}). These verification goals are sent to KeY, which starts an automatic verification attempt. The syntax of trait composition is shown in line 24. In a separate tc-file, the name of the resulting trait is given and the composed traits are connected with a plus operator.}

\begin{lstlisting}[caption={Example in our implementation},captionpos=b,label={code:exampleJava}]
public interface MaxElement1 {
/*@ requires list.size() > 0;
  @ ensures (\forall int n; list.contains(n);
  @ \result >= n) & list.contains(\result);
  @*/
  public default int maxElement(List list) {
    if (list.size() == 1) return accessHead(list);
    if (list.element() >= maxElement(list.tail())) 
      { return accessHead(list) }
    else { return maxTail(list) } }

/*@ requires list.size() > 0;
  @ ensures \result == list.element();
  @*/
  public int accessHead(List list);

/*@ requires list.size() > 1;
  @ ensures (\forall int n; list.tail().contains(n);
  @ \result >= n) & list.tail().contains(\result);
  @*/
  public int maxTail(List list);
}

ComposedMax = MaxElement1 + MaxElement2
\end{lstlisting}

\paragraph{Evaluation}
 
We evaluate our implementation by a feasibility study. First, we reimplemented an already verified case study in our trait-based language. We used the IntList~\cite{scholz2011intlist} case study, which is a small software product line (SPL) with a common code base and several features extending this code base. Here, we can show that our trait-based language also facilitates reuse.
The IntList case study implements functionality to insert integers to a list in the base version. Extensions are the sorting of the list and different insert options (e.g., front /back). We implement five methods that exists in different variants with our trait-based CbC approach.
We implement the case study in different granularities. The coarse-grained version is similar to the SPL implementation we started with~\cite{scholz2011intlist}, confirming that traits are also amenable to implement SPLs as shown by Bettini et al.~\cite{bettini2010implementing}. The fine-grained version implements the five methods incrementally with 12 refinement steps. We can reuse 6 of these steps during the construction of method variants.

We also implement three more case studies BankAccount~\cite{thuem2012familybased}, Email~\cite{hall2005fundamental}, and Elevator~\cite{plath2015elevator} with TraitCbC and CbC to show that it is feasible to implement object-oriented programs with both approaches. We used CorC~\cite{runge2019tool} as an instance of a CbC tool.
We were able to implement nine classes and verify 34 methods with a size of 1-20 lines of code. For future work, a user study is necessary to evaluate the usability of TraitCbC in comparison to CbC to confirm our stated advantages.

\section{Related Work}

Traits are introduced in many languages to support clean design and reuse, for example Smalltalk~\cite{ducasse2006traits}, Java~\cite{bono2014trait} by utilizing default methods in interfaces, and other Java-like languages~\cite{bettini2013traitrecordj,liquori2008feathertrait,smith2005chai}.
The trait language TraitRecordJ was extended to support post-hoc verification of traits~\cite{damiani2014verifying}. The authors added specifications of methods in traits for the verification of correct trait composition and proposed a modular and incremental verification process.
None of these trait languages were used to formulate a refinement process to create correct programs. They only focus on code reuse or post-hoc verification.



Automatic verification is widely used for different programming languages.
The object-oriented language Eiffel focuses on design-by-contract~\cite{meyer1988eiffel,meyer1992applying}. All methods in classes are specified with pre-/postconditions and invariants for verification purposes. The tool AutoProof~\cite{khazeev2016initial,tschannen2015autoproof} is used to verify the correctness of implemented methods. It translates methods to logic formulas, and an SMT solver proves the correctness. For C\#, programs written in the similar language Spec\#~\cite{barnett2004spec} are verified with Boogie. That is, code and specification are translated to an intermediate language and verified~\cite{Barnett:2011:SVS:1953122.1953145}. For C, the tool VCC~\cite{cohen2009vcc} reuses the Spec\# tool chain to verify programs. The tool VeriFast~\cite{jacobs2010quick} is able to verify C and Java programs specified in separation logic. For Java, KeY~\cite{ahrendt2016deductive} and OpenJML~\cite{cok2011openjml} verify programs specified with JML.
TraitCbC is parametric in the specification language, meaning that a trait-based language with a specification language and a corresponding program verifier can be used to instantiate TraitCbC. In our implementation, we use KeY~\cite{ahrendt2016deductive} to prove the correctness of methods and trait composition.

Event-B~\cite{EventB} is a related correctness-by-construction approach. In Event-B, automata-based systems are specified and refined to a concrete implementation. Event-B is implemented in the Rodin platform~\cite{abrial2010rodin}. In comparison to CbC by Kourie and Watson~\cite{kourie2012correctness} as used in this paper, Event-B works on a different abstraction level with automata-based systems instead of program code. The CbC approaches of Back et al.~\cite{back2012refinement} and Morgan~\cite{morgan-book} are also related. Back et al.~\cite{back2012refinement} start with explicit invariants and pre-/postconditions to refine an abstract program to a concrete implementation, while Kourie and Watson only start with a pre-/postcondition specification.
These refinement approaches use specific refinement rules to construct programs which are external to the programming language. With TraitCbC, we propose a refinement procedure that is part of the language by using trait composition.

Abstract execution~\cite{steinhofel2019abstract} verifies the correctness of methods with abstract, but formally specified expressions. Abstract Execution is similar to our refinement procedure where abstract methods are called in methods under construction. The difference is that abstract execution extends a programming language to use any expression in the abstract part, not only method calls. Therefore, abstract execution can better reason about irregular termination (e.g., break/continue) of methods. In comparison to TraitCbC, abstract execution is a verification-centric approach without a guideline on how to construct programs.

\hl{Synthesis of function summaries is also related~\cite{hoare1971procedures,Chen2015syn,sery2012summaries}. Here, verification tools automatically synthesize pre-/postconditions from functions to achieve modular verification and speed up the verification time. In comparison, TraitCbC is a complete software development approach where specification and code are developed simultaneously by a developer to achieve a correct solution. Function summaries are just a verification technique.}

\section{Conclusion}

In this work, we present TraitCbC that guides programmers to correct implementations. In comparison to classical CbC, TraitCbC uses method calls and trait composition instead of refinement rules to guarantee functional correctness.
We formalize the concept of a trait-based object-oriented language where the specification language is parametric to allow a broader range of languages to adopt this concept. The main advantage of TraitCbC is the simplicity of the refinement process that supports code and proof reuse.

As future work, we want to investigate how TraitCbC can be used to construct software product lines. As proposed by Bettini et al.~\cite{bettini2010implementing}, trait languages are able to implement SPLs. We want to extend the guideline of TraitCbC to construct SPLs with a refinement-based procedure that guarantees the correctness of the whole SPL.
To reduce specification effort in TraitCbC, inheritance of traits is useful as future expansion. Another option is to integrate the concept of Rebêlo et al.~\cite{Rebelo2014AspectJML} which supports the design-by-contract approach with Aspect-JML and integrates crosscutting contract modularization to reduce redundant specifications.

Since TraitCbC is parametric in the specification logic,
TraitCbC's soundness only holds if such logic is consistent when
composed in the presented manner.
In particular, the logic needs to take into account
ill-founded specifications and non-terminating recursion.
In verification, ill-founded specifications and termination
issues are often considered as a second step\footnote{For example, Dafny approximately checks that the functions used in a specification form an acyclic graph.}, separately from the
verification of individual methods, and our prototype still does not yet
take care of this second step.
That means that methods are verified under the assumption that all
other methods respect their contracts.
If ill-founded specifications and non-terminating recursion are handled
naively, verification might be unsound because of ill-founded
reasoning.
Appendix~\ref{sec:circularity} shows that this problem is even more pervasive in the
case of trait composition or any other form of multiple inheritance:
naive composition of correct traits may produce incorrect results.

\bibliographystyle{splncs03}
\bibliography{BibliographyShort}

\appendix
\section{Appendix}

\subsection{Typing Rules}
\label{sec:typing}

\begin{description}
	\item[\textsc{x}.] As usual, the type of a variable is stored in the environment $\Gamma$.
	From the verification perspective, we do not need to prove anything to be allowed to use a variable; thus we use $\mathit{true}$.
	We know that the result of evaluating a variable is the value of such variable, and that such value is of the type of the variable; thus we have $\mathit{\result:\Gamma(x)\ \&\ \result = x}$. The $\result$ is the returned value of evaluating this expression, and $\mathit{variable:type}$  is a predicate in our system.
	As you can notice, we are assuming that our parametric logic supports at least a logical \textit{and} ($\&$); but of course other ways to merge knowledge could work too.
	
	\item[\textsc{Method}.] As usual, to type a method call, we inductively type the receiver and
	all the parameters. In this way, we obtain all the types $C_0 \dots C_n$, all the knowledge
	$P_0 \dots P_n$, and all the verification obligations $P'_0 \dots P'_n$.
	Inside of all conditions  $P_i \models P'_i$ we call the result of $e_i$ $\result$.
	We cannot simply merge the knowledge of $P_0 \dots P_n$, since their
	$\result$ refers to different concepts. Thus, we chose fresh $x'_0 \dots x'_n$ variables, and we rename $\result$ of $P_i$
	and $P'_i$ into $x'_i$.
	Similar, $S'$ is the specification of the method adapted using $x'_0 \dots x'_n$.
	
	The verification obligation of course contains
	all the obligations of the receiver and the parameters, but also
	requires the precondition of the method to hold.
	
	The knowledge contains the knowledge of the receiver and the
	parameters, and the method specification in implication form.
	Naively, one could expect that since the precondition is already in
	the obligation we could simply add the postcondition to the
	knowledge.
	This would be unsound. By using the specification in implication form, the system prevents circular reasoning: we could otherwise
	use the postcondition to prove the precondition.
	Instead, when the system shows that the precondition of $S'$ holds,
	it can assume the postcondition of $S'$.
	Similar to logical \textit{and} above, we are assuming that our parametric logic supports at
	least logical \textit{implication}, but of course other forms of logical consequence could work too.
	
	Note that the postcondition will contain information about the
	result of the method body as information on the $\result$ variable.
	
	\item[\textsc{New}.] As usual, to type an object instantiation, we inductively type all the
	parameters. In this way we obtain all the types $C_1 \dots C_n$, all the knowledge
	$P_1 \dots P_n$, and all the verification obligations $P'_1 \dots P'_n$.
	As we did for \textsc{Method} we use fresh variables to be able to compose predicates.
	
	As we mentioned above, we rely on abstract state operations to represent state:
	that is, all the abstract methods in $C$ need to be of form $S_i\ \mathtt{method}\ C_i\ x_i();$
	where $\this.x_i()$ returns the value of field $x_i$, that in turn was
	initialized with the result of expression $e_i$. The function $\mathit{getters}(C)$ returns all methods of this form.
	
	Knowledge $P''_i$ contains the knowledge of $P_i$ (from expression $e_i$) and it links
	such knowledge to the result of calling method $\result.x_i()$, so that
	calling a getter on the created object will return the expected value.
	However, the information is conditional over verifying the precondition
	of such getter.
	Note that we do not need to add the knowledge of the postcondition of
	$x_i()$ here; this will be handled by the \textsc{Method} rule when $x_i()$ is called.
	
	Knowledge $P$ is simply merging the accumulated knowledge; while the final obligation
	in addition to merging the accumulated obligations also requires that the
	precondition of all the getters hold. In this way the getter
	preconditions behave like the precondition of the constructor.
	By requiring those preconditions, we ensure that we can call the
	getters on all the created objects.
	\item[\textsc{Sub}.] The subsumption rule is standard. We allow subtyping between class names. Note that we do not apply weakening and strengthening of conditions here.
\end{description}
Besides of typing correct programs, the typing rules of Trait-CbC have the goal to verify the correctness of method implementations.
The following rules check whether a method or a $\mathit{Body}$ are correct. The check for a correct method declaration in \textsc{MOK} calls a program verifier to verify the correctness. We need just one verifier call for the verification of each method because the rules above collected all needed knowledge and obligations.
\begin{description}
	\item[\textsc{MOK}.] 
	In \textsc{MOK}, we construct a $\Gamma$, and we type the method body, obtaining
	knowledge $P$ and obligation $P'$.
	The program verifier will know the type information of $\Gamma$, the premise of the method,
	and the knowledge $P$, and will prove the obligation $P'$ and the postcondition of the method. This verification in the typing rule is indicated by the keyword \textbf{verify}. 
	Here, we use implication, but a different program verifier may use a
	different form of logical consequence.
	The program verifier can access the specification of all the other methods since
	we also provide the declaration table.
	\item[\textsc{AbsMOK}.] Abstract methods are correctly typed.
	\item[\textsc{BodyTyped}.] A $\Body$ is correctly typed, if all the methods in the declaration of the $\Body$ are correctly typed.
\end{description}

\paragraph{Example on Typing}
To illustrate an example with explicit proof goals, we use the verification of the method \Q@accessHead@. In the body of \Q@accessHead@, the method \Q@list.element()@ is called. It is a getter-method with the simple postcondition \Q@result == list.element()@.
The proof obligation created by the rule \textsc{MOK} is the following. In the typing rule, the expression $e$ is the call of the method \Q@list.element()@. We use the rule \textsc{Method} to type this expression, and we obtain the following proof obligation:

$\Gamma\ \&\ $
\Q@result@
$:$
\Q@Num@
$\ \&$
\Q@list@
$ = x'_0\ \&\ (\mathit{true}$
$ \Rightarrow $
\Q@result@
$ = x'_0$
\Q@.element()@
$)\ \models  $
\Q@result@
$ = $
\Q@list.element()@

We prove that the body of \Q@accessHead@ satisfies the pre-/postcondition specification. The precondition of the called method \Q@list.element()@ is satisfied, as it is $\mathit{true}$. Therefore, we can use the postcondition of this method to show that the postcondition of \Q@accessHead@ holds. If we replace $x'_0$ with \Q@list@, we have the same condition in our pre- and postcondition and can close the goal. Thus, the method \Q@accessHead@  is proven correct.

\subsection{Reduction Rules}
\label{sec:reduction}

We formulate three reduction rules for our system to evaluate input expressions to final values. We introduce an evaluation context $\E_v$ in our syntax in Fig.~\ref{f:syntax} to define the order of evaluation. The rules of Fig.~\ref{f:reduction} are explained in the following.

\begin{figure}[t]
	\centering
	\begin{scriptsize}
		\begin{mathpar}
			\inferrule
			{
				\Ds \vdash e \rightarrow e'
			}{
				\Ds \vdash \E_v[e] \rightarrow \E_v[e']}
			\quad (\textsc{$Ctx$})
			
			\inferrule
			{
				S\ \methodKW\ C\ m(C_1\ x_1,\ \dots,\ C_n\ x_n)\ e;\ \in\ \mathit{methods(C)}
			}{
				\Ds \vdash \newKW\ C(\vs).m(v_1\dots v_n) \rightarrow e[\this=\newKW\ C(\vs),\ x_1=v_1,\ \dots,\ x_n=v_n]}
			(\textsc{mcall})
			
			\inferrule
			{
				\mathit{abs}(\Ds(C))=S_1\ \methodKW\ C_1\ x_1();\ \dots\ S_n\ \methodKW\ C_n\ x_n();
			}{
				\Ds \vdash \newKW\ C(v_1 \dots v_n).x_i() \rightarrow  v_i}
			\quad (\textsc{getter})
		\end{mathpar}
	\end{scriptsize}
	\caption{Reduction rules}
	\label{f:reduction}
\end{figure}

\begin{description}
	\item[\textsc{Ctx}.] This is the conventional contextual rule, allowing the execution of subexpressions.
	\item[\textsc{Mcall}.] We reduce a method call to an expression $e$, where the receiver is replaced with \newKW\ $C(\vs)$, and each parameter $x_i$ with the actual value $v_i$.
	We also ensure that the method is declared in the class $C$.
	\item[\textsc{Getter}.] In our formalism, abstract methods without arguments represents getters. Notation $\mathit{abs(Body)}$ returns the set of all abstract methods in $\Body$. A valid class can only have
	abstract methods without arguments, and they will all represent getters.
\end{description}

\subsection{Proofs}
\label{sec:proofs}

Proof of Lemma~\ref{lemma:flattening}.

\begin{proof}
	We prove by contradiction. 
	We assume the resulting $\Body$ is ill typed.
	By definition of \textsc{BodyTyped}, it means that one of the methods cannot be
	typed with either \textsc{AbsMOK} or \textsc{MOK}.
	The list of methods that need to be typed is obtained by Definition~\ref{def:bodyplus}.
	
	Abstract methods can only be typed with \textsc{AbsMOK} and are never wrong.
	Implemented methods can only be typed with \textsc{MOK}.
	If $\Gamma \vdash e : C \dashv P \models P'$
	or the other precondition $\textbf{verify}\ \Ds \vdash (\Gamma\ \&\ \Pre(S)\ \&\ P) \models (P'\ \&\ \Post(S))$ does not hold, 
	it means that there was a method $m_i$ with expression $e_i$ in $\Body_1$ (or
	symmetrically for $\Body_2$)
	that was well-typed under $\Ds; t_1\vdash \Body_1$. That means that all of its implemented methods were well-typed and verified.
	Typing $e_i$ produces
	$P_i \models P'_i$ by using a $\Gamma_{t1}$ containing $\this:t_1$.
	
	If $\Ds; \Name \vdash \Body : \OK$ is not applicable, the same
	expression $e_i$
	was typed using a $\Gamma_{\Name}$ containing $\this:\Name$. It produced $P''_i \models P'''_i$
	so that $\textbf{verify}\ \Ds \vdash (\Gamma_{\Name}\ \&\ \Pre(S)\ \&\ $ $P''_i) \implies (P'''_i\ \&\ \Post(S))$ does not hold.
	We know that $\textbf{verify}\ \Ds \vdash (\Gamma_{t1}\ $ $\&\ \Pre(S)\ \&\ P_i) \implies (P'_i\ \&\ \Post(S))$ holds by our assumption.
	By Definition~\ref{def:mplus}, the contracts of the methods in $\Body$ are simply stronger than the contracts of the methods in $\Body_1$.
	The only difference between $P''_i \models P'''_i$ and $P_i \models P'_i$ is in
	the contracts of methods called on $\this$.
	Assuming that our parametric logic implication is transitive,
	we know that $\textbf{verify}\ \Ds \vdash (\Gamma_{t1}\ \&\ \Pre(S)\ \&\ P_i) \implies (P'_i\ \&\ \Post(S))$ entails $\textbf{verify}\ \Ds \vdash (\Gamma_{\Name}\ \&\ \Pre(S)\ $ $\&\ P''_i) \implies (P'''_i\ \&\ \Post(S))$,
	thus we reach a contradiction.
\end{proof}

\noindent Proof of Lemma~\ref{lemma:makeAbstract}.

\begin{proof}
	We prove by contradiction. 
	We assume the resulting $\Body'$ is ill typed.
	By definition of \textsc{BodyTyped}, it means that one of the methods cannot be
	typed with either \textsc{AbsMOK} or \textsc{MOK}.
	The list of methods that need to be typed is obtained by Definition~\ref{def:bodyplus}.
	
	Abstract methods can only be typed with \textsc{AbsMOK} and are never wrong.
	We know that $\Body$ is typable by our assumption. The only difference between $\Body$ and $\Body'$ is that the method $m$ is made abstract. As we have seen for Lemma~\ref{lemma:flattening}, we are typing $\Body'$ in a different $\Gamma$. This case is even simpler than Lemma~\ref{lemma:flattening} because $\Body$ and $\Body'$ have exactly the same specifications. The abstract method $m$ and thus $\Body'$ cannot be ill typed.
\end{proof}

\noindent Proof of Theorem~\ref{theorem:soundness}.

\begin{proof}
	By induction on the size of $\Ds$, and by induction on cases of $E$ (the
	applied flattening rule for $E$).
	\begin{itemize}
		\item $\Body$ only flattens if the $\Body$ can be shown to be well-typed.
		\item $t$ only reads a trait from the already verified $\Ds'$.
		\item $\Body_1+\Body_2$ is correct with Lemma~\ref{lemma:flattening}. The lemma can be applied directly, if $E$ is of depth one (e.g., $\Body_1 + \Body_2$). If $E$ is more complex, we have to apply other cases of this case analysis.
		\item $\mathtt{makeAbstract}$ is handled similarly using Lemma~\ref{lemma:makeAbstract}.
		\item By the flattening relation, we know that $\Body_1$ and $\Body_2$ are well-typed in $\Ds$.
		If we start from a program containing only well-typed and verified traits,
		any new class we can define by just composing those traits is well
		typed and verified.
	\end{itemize}
\end{proof}

\subsection{Circularity Issue}
\label{sec:circularity}
A major problem of verification is the occurrence of circularity in the specification of methods~\cite{ahrendt2016deductive}. In the following code snippet, we have two traits. The first trait $A$ implements a method $a$ and requires a method $b$. The second trait $B$ implements a method $b$ and requires a method $a$. Both traits on their own can be verified (It will be assumed that the abstract methods satisfy their contract).

\begin{lstlisting}
trait A {
  @Post: result = b()
  Num a() = 5

  @Post: result = 5
  Num b();
}

trait B {
  @Post: result = 5
  Num a();

  @Post: result = a()
  Num b() = 5
}
\end{lstlisting}

If we could compose both traits, the flattened program would contain the implementation of $a$ and $b$ with the contracts of the implemented ones. 
Even if we modularly verify the methods, we can only prove $a$ under the assumption that $b$ is correct, and $b$ under the assumption that $a$ is correct, which again is circular.
The result of method $a$ is equal to the result of $b$, but the result of $b$ is specified as the result of $a$. 
This verification problem is well-known~\cite{ahrendt2016deductive} but it cannot be solved easily. With appropriate analysis techniques, loops in code or specification can be detected and marked for the programmer. In our proposed language, we can construct examples, such as the one we presented, where the circle is obscured by trait composition. This makes reasoning very challenging.

\end{document}